[2023/10/19]
\documentclass[draft]{article}
\title{Spectral theory of $p$-adic unitary operator}
\author{Zhao Tianhong\\{\footnotesize University of Science and Technology of China}\\{\footnotesize Email: \texttt{ustczth@mail.ustc.edu.cn}}}
\date{2023/10/19}

\usepackage{amsfonts}
\usepackage{amsmath}
\usepackage{amssymb}
\usepackage{mathtools}
\usepackage{amsthm}
\usepackage{pgfplots}
\usepackage{enumitem}
\usepackage{tikz} 
\usepackage{tikz-cd}
\usepackage{changepage}

\chardef\bslash=`\\ 

\newtheorem{thm}{Theorem}[section]
\newtheorem{cor}[thm]{Corollary}
\newtheorem{lem}[thm]{Lemma}
\newtheorem{prop}[thm]{Proposition}
\newtheorem{ax}{Axiom}

\newtheorem{ex}{Example}[section]

\theoremstyle{definition}
\newtheorem{defn}{Definition}[section]

\newtheorem{q}{Question}[section]

\theoremstyle{remark}
\newtheorem{rem}{Remark}[section]
\newtheorem*{notation}{Notation}

\newcommand{\R}{\mathbb{R}}
\newcommand{\Z}{\mathbb{Z}} 
\newcommand{\N}{\mathbb{N}}
\newcommand{\Q}{\mathbb{Q}}
\newcommand{\CC}{\mathbb{C}}
\newcommand{\G}{\mathbb{G}_m}

\newcommand{\Fp}{\mathbb{F}_p}
\newcommand{\GL}{\operatorname{GL}} 
\newcommand{\Spec}{\operatorname{Spec}} 
\newcommand{\Hom}{\operatorname{Hom}} 
\newcommand{\Berk}{\operatorname{Berk}} 
\newcommand{\Max}{\operatorname{Max}} 
\newcommand{\Aff}{\operatorname{Aff}} 
\newcommand{\Rt}{\operatorname{Root}} 
\newcommand{\End}{\operatorname{End}} 
\newcommand{\Fct}{\operatorname{Fct}} 
\newcommand{\Comp}{\operatorname{Comp}} 
\newcommand{\coker}{\operatorname{coker}} 
\newcommand{\Gal}{\operatorname{Gal}} 
\newcommand{\Sp}{\operatorname{Sp}} 
\newcommand{\AB}{\textbf{$A-\textup{Banach}$}} 
\newcommand{\GB}{\textbf{$\mathbb{G}_m-\textup{Banach}$}} 
\newcommand{\OA}{\textbf{$\mathcal{O}_p-\textup{Banach Alg}$}}

\begin{document}
\maketitle
\begin{abstract}
	The $p$-adic unitary operator $U$ is defined as an invertible operator on $p$-adic ultrametric Banach space such that $\left |U\right |=\left |U^{-1}\right |=1$. We point out $U$ has a spectral measure valued in $\textbf{projection functors}$, which can be explained as the measure theory on the formal group scheme. The spectrum decomposition of $U$ is complete when $\psi$ is a $p$-adic wave function. We study $\textbf{the Galois theory of operators}$. The abelian extension theory of $\Q_p$ is connected to the topological properties of the $p$-adic unitary operator. We classify the $p$-adic unitary operator as three types: \textbf{Teichmüller type}, \textbf{continuous type}, \textbf{pro-finite type}. Finally, we establish a $\textbf{framework of $p$-adic quantum mechanics}$, where projection functor plays a role of quantum measurement.
\end{abstract}

\section{Introduction}
This paper aims to establish the spectral theory of the $p$-adic unitary operator and construct a framework for $p$-adic quantum mechanics. We apply the number theory on functional analysis to connect several areas. The physical motivation is to find the invariant structure in quantum mechanics by changing the base field. The previous work \cite{zth} of $p$-adic Hermite operator gives a table which compares the concept in $p$-adic functional analysis with the usual functional analysis over $\R,\CC$.
\begin{table}[!htbp]
	\centering
	\begin{tabular}{|c|c|c|} 
		\hline
		& Archimedean & Non-Archimedean \\ 
		\hline
		base field & $\R,\CC$ & $K$,$K$ is a extension of $\Q_p$\\
		\hline
		Banach space & Hilbert space & ultrametric Banach space\\
		\hline
		norm & $|x+y|^2+|x-y|^2=2(|x|^2+|y|^2)$ & $|x+y|\le max(|x|,|y|)$\\ 
		\hline
		orthogonal projection & $|x|^2=|\pi(x)|^2+|\pi^{\perp}(x)|^2$ & $|x|=max(|\pi(x)|,|\pi^{\perp}(x)|)$ \\ 
		\hline
		Banach algebra & $C^{*}$-Algebra & ultrametric Banach algebra\\
		\hline
		Galois action & Hermite conjugate $\dagger$ & Frobenius map $\sigma$\\
		\hline
	\end{tabular}
\end{table}

The spectral theory, which involves the conventional Hermite operator and unitary operator, serves as a conceptual tool in quantum mechanics. We suggest a categorical point of view that substitutes the orthogonal projection with the projection functor. This implies that the language of quantum mechanics could extend beyond operators to include functors.

The spectral measure of operators comes from measure theory on group scheme, which can be understood as a functorial measure theory. Moreover, the structures of base field has a deep influence on the properties of operators. This leads us to study the Galois theory of operators. The relationship between Galois groups and operators could be listed as follows:
\begin{table}[!htbp]
	\centering{
		\begin{tabular}{|c|c|c|} 
			\hline
			field & Galois group & operator \\ 
			\hline
			$\R,\CC$ & $\Gal\left(\CC\mid \R\right)$ & Hermite operator \\
			\hline
			$\R,\CC$ & $\Gal\left(\CC\mid \R\right)$ & unitary operator   \\
			\hline
			$\Q_p,\Q_p^{ur}$ & $\Gal\left(\Q_p^{ur}\mid \Q_p\right)$ & $p$-adic Hermite operator \\ 
			\hline			
			$\Q_p,\Q_p^{ab}$ & $\Gal\left(\Q_p^{ab}\mid \Q_p\right)$ & $p$-adic unitary operator \\ 
			\hline
		\end{tabular}
	}
\end{table}

We refer to \cite{Acourse} for the basic concepts in $p$-adic analysis. We refer to \cite{NAUO} \cite{NAOA} \cite{PO} for the recently progress in $p$-adic functional analysis. The connection between analytic geometry and spectral theory is discussed by \cite{Berk}. We refer to \cite{FSFG} for theory of formal scheme and formal group. 
We refer to \cite{pMP} \cite{pQM} for previous work in $p$-adic quantum mechanics. Noted that the $p$-adic wave function take values in usual complex field $\CC$ in their papers. The wave function in this paper take values in $\CC_p$.

\section{Definition}
\begin{notation}
	
	\begin{description}
		\item 
		\item[$\CC_p$] $p$-adic complex field.
		\item[$\mathcal{O}_p$] Integral ring of $\CC_p$.
		\item[$\mathfrak{m}_p$] Maximum ideal of $\mathcal{O}_p$.
		\item[$\G$] Formal group scheme over $\mathcal{O}_p$.
		\item[$A$] strict ultrametric Banach algebra.
		\item[$\R_+$] The internal $\left [0,+\infty\right )\subset \R$
		\item[$\varepsilon$] $p$-adic number in $\mathfrak{m}_p$ which is used to do reduction. 
		\item[$\underset{\varepsilon \in \mathfrak{m}_p^{+}}{\varprojlim} A/\varepsilon  A$] The completion of $A$ in $p$-adic topology.
		\item[$\Comp A$] The set of open ideal of $A$.
		\item[$\Pi_{I}$] The left projection functor with respect to open ideal $I$.
		\item[$\pi_{I}$] The right projection functor with respect to open ideal $I$.
		\item[$\Q_p^{ur}$] The maximum unramified extension of $\Q_p$.
		\item[$\Q_p^{ab}$] The maximum abelian extension of $\Q_p$.
		\item[$\Gal\left(\CC_p \mid \Q_p\right) $] The Galois group of extension $\CC_p \mid \Q_p$ consists of the continuous $\Q_p$-isomorphisms.
	\end{description}
\end{notation}

Let $\CC_p$ be the $p$-adic complex numbers, $\left |-\right |_p$ be the $p$-adic norm over $\CC_p$, $\mathcal{O}_p$ be the ring :
$$
\mathcal{O}_p=\left \{ x\in \CC_p,\left | x \right |_p\le 1  \right \}.
$$
Let $\G$ be the $\mathcal{O}_p$-Banach algebra: 
$$
\G: =\left \{ f(t)=\sum_{n=-\infty}^{\infty}a_nt^n,a_n \in \mathcal{O}_p,a_n \to 0,\ n \to \pm \infty \right \}. \\
$$
with the norm $\left | f(t) \right |_{\G} =\sup_{n\in \Z}\left |a_n\right |_p $.
Let $X$ be a $\CC_p$-ultrametric Banach space, $\left |-\right |_X$ be the norm. Let $M$ be the unit ball:
$$
M: =\left \{ x\in X,\left | x \right |_X \le 1  \right \} .
$$
\begin{defn}
Let $U: X \to X$ be a bounded invertible $\CC_p$-linear operator, we say $U$ is a $p$-adic unitary operator if one of the following equivalent conditions holds:
\begin{align*}
&1.U(X_M)=X_M.\\
&2.\left | Ux \right |_M=\left |x \right |_M, \forall x \in X_M.\\
&3.\left | U \right | =\left | U^{-1} \right | =1 ,\textup{where} \left |- \right |\textup{ is operator norm}.
\end{align*}
\end{defn}
It's equivalent to say that a $p$-adic unitary operator $U$ gives $M$ a $\G$-module structure, where the action of $\G$ on $M$ is given by:
\begin{align*}
\G \times M &\longrightarrow  M\\
\left (  \sum_{n=-\infty}^{\infty}a_nt^n ,x\right )  &\longmapsto    \sum_{n=-\infty}^{\infty}\left ( a_nU^n \left ( x\right )\right ) .
\end{align*}
So the spectral theory of $p$-adic unitary operator is dominated by $\G$. In scheme-theoretic sense $\G$ is a formal group scheme. Let $R$ be an arbitary $\mathcal{O}_p$-ultrametric Banach algebra ( not necessary commutative ), $\Hom(\G,R)$ be the set of contraction morphism of $\mathcal{O}_p$-Banach algebra. We have:
$$
\Hom(\G,R) \overset{1:1}{\longleftrightarrow} U(R)=\left \{ x\in R,\left | x \right |=\left | x^{-1} \right |=1  \right \}
$$
where $U(R)$ is group consisted of all $p$-adic unitary elements in $R$.\\
Following the Grothendieck's philosophy, $\Hom(\G,-)$ defines a functor:
\begin{align*}
\OA &\longrightarrow \textbf{Group} \\
R  &\longmapsto  U(R) .
\end{align*}
The spectral theory of $p$-adic unitary operator could be explained as \textbf{measure theory on formal group scheme $\G$} later.
\section{Berkovich space and Gelfand representation}
Let $A$ be a $\CC_p$-ultrametric Banach algebra with a norm $\left |- \right |_A$. 
\begin{defn}
We say $\left |-\right |:A \to \R_+ $ is a multiplicative semi-norm, if $\left |-\right |$ satisfy the following properties:
\begin{enumerate}
\item $\left |cx\right |= \left |c\right |_p\left |x \right |,\forall c\in \CC_p,x \in A;$
\item $\left |x+y\right |\le \sup(\left |x \right |,\left |y \right |),\forall x,y \in A;$
\item $\left |xy\right |= \left |x\right |\left |y \right |,\forall x,y \in A.$
\end{enumerate}
\end{defn}

\begin{defn}
The Berkovich space of $A$ is the set of multiplicative semi-norm bounded by $\left |- \right |_A$:
$$
\Berk A=\left \{\rule{0mm}{4.3mm}\left | - \right |:A \to \R_+  \;\middle|\; \left | - \right |\ \text{is multiplicative semi-norm},\left | - \right |\le \left | - \right |_A  \right \}
$$
\end{defn}
\begin{lem}(Monotone Convergence Theorem)
For
$$\left | - \right |_1\le \left | - \right |_2 \le \left | - \right |_3 \le \cdots, \ \left \{\rule{0mm}{4.3mm}\left | - \right |_i   \right \}_{i=1,2,3...}\subseteq \Berk A$$
then we have:
$$
\sup_{i=1,2,3...} \left ( \left | - \right |_i  \right ) \in \Berk A.
$$
\end{lem}
From Zorn's lemma, there exists some maximum norms. Maximum norms of Berkovich space are correspond to \textbf{generic points} or \textbf{irreducible components} in scheme theory. We can define a canonical spectral semi-norm on $A$.
\begin{defn}
Let $x \in A$, define the \textbf{spectral semi-norm} $\left |-\right |_{sp}$:
$$
\left |x\right |_{sp}:=\sup_{\left |-\right |\in \Berk A} \left |x\right |
$$
\end{defn}
Let $A_{sp}$ be the completion of $A$ over the spectral semi-norm $\left |-\right |_{sp}$, the morphism: $\Gamma:A \to A_{sp}$ is a generalized Gelfand representation map. However, it is not an isometry in general cases. In general, we have:
$$
\left |x\right |_{sp} \le \left |x\right |_A, \forall x \in A.
$$
Let $A$ be a commutative $C^*$-Algebra with unit, 
We say $\left |-\right |:A \to \R_+ $ is a multiplicative semi-norm, if $\left |-\right |$ satisfy the following properties:
\begin{enumerate}
\item $\left |cx\right |= \left |c\right |\left |x \right |,\forall c\in \CC,x \in A;$
\item $\left |x+y\right |\le \left |x \right |+\left |y \right |,\forall x,y \in A;$
\item $\left |xy\right |= \left |x\right |\left |y \right |,\forall x,y \in A.$
\end{enumerate}
the Berkovich space of $A$ is defined as:
$$
\Berk A=\left \{\rule{0mm}{4.3mm}\left | - \right |:A \to \R_+  \;\middle|\; \left | - \right |\ \text{is multiplicative semi-norm},\left | - \right |\le \left | - \right |_A  \right \}
$$
We have:
$$
\begin{tikzcd}
\Berk A &\overset{1:1}{\longleftrightarrow } &\Max A &\overset{1:1}{\longleftrightarrow } &\Hom(A,\CC) \\
\left |-\right |&\overset{1:1}{\longleftrightarrow } & \ker\left |-\right | &\overset{1:1}{\longleftrightarrow } &  A/\ker\left |-\right | 
\end{tikzcd}
$$

where $\Max A$ is the set of all maximum ideal of A, $\Hom(A,\CC)$ is the set of all $\CC$-Banach algebra morphism\footnote{Some papers call $\Hom(A,\CC)$ characters.}. We can define a Gelfand topology on $\Max A$, which is the weakest topology such that all elements $f \in A$ be a continuous function on $\Max A$. Let $C(\Max A)$ be the set of complex valued continuous function on $\Max A$, the Gelfand representation gives a isometry isomorphism: $A \overset{\sim}{\to} C(\Max A)$. We have:
$$
\left |x\right |_{sp} =\sup_{\lambda \in \Max A}\left ( \left |x(\lambda)\right |\right )= \left |x\right |_A, \forall x \in A.
$$
\begin{defn}
Let $U$ be an invertible operator on a Hilbert space $\mathcal{H}$ over $\CC$, we say $U$ is a unitary operator if the following condition holds:
$$
U^{\dagger}=U^{-1}.
$$
\end{defn}
Let $U$ be a usual unitary operator act on a Hilbert space $\mathcal{H}$ over $\CC$, $\End \mathcal{H}$ be the set of bounded operator on $\mathcal{H}$. Let 
$$p(U)= \left \{\rule{0mm}{4.3mm}f(U)\in \End \mathcal{H},f(U)= \sum_{k\in \Z} c_kU^k, c_k\in \CC \right \}$$
be the set of polynomial of $U$ consists of positive degree and negetive degree. $p(U)$ has a norm structure from the operator norm defined on $\End \mathcal{H}$. Let $A$ be the completion of $p(U)$. $A$ is the $C^*$-algebra generated by the $U$, the spectral theory of usual unitary operator $U$ in Gelfand sense is exactly the spectral theory of $A$,we have:
$$A \simeq C(K).$$ 
$K=\Max A$ is the spectrum of $U$, which is a compact subset of the unit cycle $S^1$. We can generalize this statement as the following theorem:
\begin{thm}
(Spectral theory of usual unitary operator)
The complex valued contious function space $C(S^1)$ is a group scheme in the category of unital non-commutative $C^*$-algebra. Let $R$ be a $C^*$-algebra.
Let 
$$
U(R)=\left \{ x\in R,x^{-1}=x^{\dagger}  \right \}
$$
be the set of unitary elements of $R$, $U(R)$ is a group. $\Hom(C(S^1),-)$ defines a functor:
\begin{align*}
\textbf{$C^*$-\textup{Algebra}} &\longrightarrow \textbf{\textup{Group}} \\
R  &\longmapsto U(R).
\end{align*}
Let $K_1,K_2$ be the compact set in $S^1$, then $K_1\cup K_2$ is a compact set in $S^1$, so the compact set in $S^1$ defines a direct system, ordered by relation of inclusion. We have:
\begin{align*}
&1.S^1=\bigcup_{K \subset S^1}K=\underset{K \subset S^1}{\varinjlim}K,\\
&2.C(S^1)=\underset{K \subset S^1}{\varprojlim}C(K),\\
&3.C\left ( K_1\coprod K_2\right ) =C(K_1)\times C(K_2).
\end{align*}
The spectral measure theory of unitary operator $U$ can be viewed as a functorial measure theory.
\end{thm}
The orthogonal projection can be generated from the representation theory of $C(S^1)$, which is a technical problem in functional analysis. Let $L$ be a compact set in $S^1$, which is a measurable set. There exists an orthogonal projection $\pi_L \in \End \mathcal{H}$ such that:
$$
\pi_L U=U\pi_L,
$$
where $\pi_L U$ is a operator whose spectrum lies in $L$. Using some technical method in functional analysis and representation theory\footnote{Riesz representation theorem is needed.}, we can define the spectral integral in classical theory of unitary operator:
$$
I=\int_{S^1}d\pi_{\lambda }  \ \ \   U=\int_{S^1}\lambda d\pi_{\lambda }
$$

In $p$-adic case, the Gelfand representation of the ultrametric Banach algebra may not be isometry. 
\begin{ex}
Let $U=\begin{pmatrix}
 1 & 1\\
 0 & 1
\end{pmatrix} \in \GL_2(\mathcal{O}_p)$, the $\CC_p$-ultrametric Banach algebra generated by $U$ is isomorphic to 
$A=\CC_p\left [X\right ]/(X-1)^2 $, the Gelfand representation of $A$ is isomorphic to $A_{sp}=\CC_p\left [X\right ]/(X-1)\simeq \CC_p$.

\end{ex}
\section{Projection functor}
\begin{notation}
	The symbol $\mathfrak{m}_p^{+}$ is defined as the set:
	$$
	\mathfrak{m}_p^{+}= \left (\mathfrak{m}_p\cup \left \{1^{-}\right \}\right )-\left \{0\right \}.
	$$
	We always use $\varepsilon \in \mathfrak{m}_p^{+}$ to do reduction, where $\varepsilon=1^{-}$ means the reduction over $\overline{\Fp}$.
	Moreover, every ring has an identity.
\end{notation}
In algebraic geometry, localization is a powerful way to study modules. Let $A$ be a commutative ring with unit, $\mathfrak p$ be a prime ideal, $A_{\mathfrak p}$ be the localization at $\mathfrak p$. $M$ be an $A$-module, we have:
$$
M=0 \iff A_{\mathfrak p}\underset{A}{\otimes} M= M_{\mathfrak p}=0, \forall \mathfrak p \in \Spec A.
$$
One can define the localization functor:$\ A_{\mathfrak p}\underset{A}{\otimes}(-)$. $A_{\mathfrak p}\underset{A}{\otimes}(-)$ is an exact functor from category $A-mod$ to $A-mod$.

When $A$ is a ultrametric commutative Banach algebra, we cannot define the well-behaved localization functor for $A$-Banach modules. However, it is possible to define the \textbf{projection functor}.
\begin{defn}
Let $A$ be a $\mathcal{O}_p$-algebra with respect to a norm $\left |-\right |_A:A \to \R_+$, we say $A$ is a \textbf{ultrametric $\mathcal{O}_p$-Banach algebra}, if the following conditions hold:
\begin{align*}
&1.\left | cx\right |_A \le \left | c \right |_p\left | x \right |_A, \forall c \in \mathcal{O}_p,x\in A; \\
&2.\left | xy \right |_A \le \left | x \right |_A \left |  y\right |_A ,\forall x,y \in A; \\
&3.\left |x+y \right |_A \le \sup( \left |x\right |_A, \left |y\right |_A ), \forall x,y\in A;\\
&4.\left(A,\left |-\right |_A\right) \textup{is complete}.
\end{align*}
We say $A$ is a \textbf{bounded ultrametric $\mathcal{O}_p$-Banach algebra} if the following condition hold:
$$
\forall x \in A, \left | x \right |_A \le 1.
$$
\end{defn}

\begin{defn}
    Let $A$ be a bounded ultrametric $\mathcal{O}_p$-Banach algebra. Suppose  $\forall \varepsilon \in \mathfrak{m}_p,\ \varepsilon  A$ is a open-closed ideal in $A$. Let's define:
$$I_{\varepsilon}=\begin{cases}
	\varepsilon  A  & \ \ \varepsilon  \in \mathfrak{m}_p \\
	\bigcup_{\varepsilon \in \mathfrak{m}_p}\varepsilon  A    & \ \  \varepsilon=1^{-}.
\end{cases}$$
    We have:
$$
A= \varprojlim_{\varepsilon \in \mathfrak{m}_p^{+}} A/I_{\varepsilon}:=\varprojlim_{\varepsilon \in \mathfrak{m}_p^{+}} A/\varepsilon A.
$$
	We say $A$ is a \textbf{strict ultrametric $\mathcal{O}_p$-Banach algebra} if the condition above holds.
\end{defn}
\begin{ex}
	$\G$ is a strict ultrametric $\mathcal{O}_p$-Banach algebra. 
\end{ex}
\begin{ex}
Let $A$ be a ultrametric $\CC_p$-Banach algebra, 
$$
\mathcal{O}_A:=\left \{ x\in A,\left |x \right |_A\le 1  \right \}  
$$
is a strict ultrametric $\mathcal{O}_p$-Banach algebra. $\forall \varepsilon \in \mathfrak{m}_p^{+}, \mathcal{O}_A/\varepsilon\mathcal{O}_A$ is a strict ultrametric $\mathcal{O}_p$-Banach algebra. 
\end{ex}

\begin{defn}
Let $M$ be a $\mathcal{O}_p$-module with respect to a norm $\left |-\right |_M:M \to \R_+$, we say $M$ is a \textbf{ultrametric $\mathcal{O}_p$-Banach module}, if the following conditions hold:
\begin{align*}
&1.\left | cx\right |_M \le \left | c \right |_p\left | x \right |_M, \forall c \in \mathcal{O}_p, x\in M; \\
&2.\left |x+y \right |_M \le \sup( \left |x\right |_M, \left |y\right |_M ), \forall x,y\in M; \\
&3.\left(M,\left |-\right |_M\right) \textup{is complete}.
\end{align*}
We say $M$ is a \textbf{bounded ultrametric $\mathcal{O}_p$-Banach module} if the following condition hold:
$$
\forall x \in M, \left | x \right |_M \le 1.
$$
\end{defn}
\begin{defn}
	Let $M$ be a bounded ultrametric $\mathcal{O}_p$-Banach module. Suppose $\forall \varepsilon \in \mathfrak{m}_p,\ \varepsilon  M$ is a open-closed submodule in $A$. Let's define:
	$$\varepsilon M=\begin{cases}
		\varepsilon  M  & \ \ \varepsilon  \in \mathfrak{m}_p \\
		\bigcup_{\varepsilon \in \mathfrak{m}_p}\varepsilon M    & \ \  \varepsilon=1^{-}.
	\end{cases}$$
	We have:
	$$
	M= \varprojlim_{\varepsilon \in \mathfrak{m}_p^{+}} M/\varepsilon M
	$$
	We say $M$ is a \textbf{strict ultrametric $\mathcal{O}_p$-Banach module} if the condition above holds.
\end{defn}

\begin{defn}
Let $\left ( A,\left |-\right |_A  \right )$ be a ultrametric $\mathcal{O}_p$-Banach algebra, $\left ( M,\left |-\right |_M  \right )$ be a ultrametric $\mathcal{O}_p$-Banach module, we say $M$ is a ultrametric $A$-Banach module, if the following condition holds:
\begin{align*}
&1. M \textup{ is an }A-mod; \\
&2.\left |am\right |_M\le \left |a\right |_A\left |m\right |_M, \forall a \in A,m \in M;
\end{align*}
\end{defn}
The morphism between bounded ultrametric $A$-Banach module is defined by the contraction morphism of $A$-mod:
$$
\Hom_{\AB}(M,N)=\left \{ \phi\in \Hom_{A-\textup{mod}}(M,N), \left |\phi(x)\right |_{N}\le \left |x\right |_{M}     \right \}. 
$$
The morphism between bounded ultrametric $\mathcal{O}_p$-Banach algebra is defined by the contraction morphism of $\mathcal{O}_p$-algebra:
$$
\Hom_{\OA}(A,R)=\left \{ \phi\in \Hom_{\mathcal{O}_p\textup{-algebra}}(A,R), \left |\phi(x)\right |_{R}\le \left |x\right |_{A}     \right \}. 
$$
Let $\AB$ donate the category of bounded ultrametric $A$-Banach module and contraction morphism. Let $\OA$ donate the category of bounded ultrametric $\mathcal{O}_p$-Banach algebra and contraction morphism. Let $\AB_S$ donate the category of strict ultrametric $A$-Banach module and contraction morphism. Let $\OA_S$ donate the category of strict ultrametric $\mathcal{O}_p$-Banach algebra and contraction morphism.
\begin{rem}
The closed $A$-submodule $N$ of bounded ultrametric $A$-Banach module $M$ is bounded ultrametric $A$-Banach module with respect to the norm:
$$
\left | x \right |_{N}=\left | x \right |_{M}
$$
\end{rem}
\begin{rem}
Suppose $N$ is a open-closed submodule of strict ultrametric $A$-Banach module $M$, the quotient module $M/N$ is strict ultrametric $A$-Banach module with respect to the quotient norm:
$$
\left | x \right |_{M/N} =\sup_{y\in N}(\left |x+y\right |_{M})  
$$
\end{rem}
\begin{defn}
Let $A$ be a ultrametric $\mathcal{O}_p$-Banach algebra, $A$ has a $p$-adic topology given by $\left |-\right |_A $. Let $I$ be a ideal of $A$. We say $I$ is \textbf{open ideal} if $I$ is open set.
\end{defn}
\begin{rem}
$I$ is open ideal if and only if: 
$$\exists \varepsilon \in \mathcal{O}_p, \varepsilon \ne 0, w(\varepsilon)  \in I,$$
where $w$ is the unit morphism: 
$$
w:\mathcal{O}_p \to A
$$
Let $I\le A$ be a open ideal, it can be showed that $I$ is closed in $A$.
\end{rem}

Let's define the projection functor.
\begin{defn}
Let $A$ be a strict ultrametric $\mathcal{O}_p$-Banach algebra, $I$ be a open ideal of $A$. The \textbf{left projection functor} $\Pi_{I}:\AB \to \AB$ is defined as:
$$
\Pi_{I}(M)=\Hom_{\AB}(A/I,M)=\left \{ x\in M  \;\middle|\; ax=0,\forall a\in I \right \}.
$$
The \textbf{right projection functor} $\pi_{I}:\AB_{S} \to \AB_{S}$ is defined as:
$$
\pi_{I}(M)=A/I\underset{A}{\otimes} M=M/IM.
$$
Since $I$ is open ideal, it is easy to show that $IM$ is open-closed submodule of $M$. Moreover, \textbf{it is convenient to permit $I=(0)$ or $(1)$}.
\end{defn}
\begin{rem}
The quotient algebra $A/I$ can be viewed as the $p$-adic counterpart of complex valued function $C(K)$ over compact set $K$. $\Pi_{I}(M)$
 gives a "categorical submodule of $M$", $\pi_{I}(M)$
 gives a "categorical quotient module of $M$". 
 Let $\mathcal{H}$ be a Hilbert space, the close subspace $X$ of $\mathcal{H}$ has a orthogonal projection theorem:
 $$
 \mathcal{H}=X \hat{\oplus}  X^{\perp}, \mathcal{H}/X \simeq X^{\perp}.
 $$ 
\end{rem}

\begin{prop}
Let $I,J$ be open ideal of $A$, $M\in \AB_S$. the projection functor has the following pullback and pushout diagram:
$$
\begin{tikzcd}
 \Pi_{I+J}(M) \arrow[hookrightarrow,r] \arrow[hookrightarrow,d] & \Pi_{J}(M) \arrow[hookrightarrow,d]  &	\pi_{I\cap J}(M) \arrow[two heads,r] \arrow[two heads,d] & \pi_{J}(M) \arrow[two heads,d] \\
 \Pi_{I}(M) \arrow[hookrightarrow,r] & \Pi_{I\cap J}(M) &\pi_{I}(M) \arrow[two heads,r] & \pi_{I+J}(M)
\end{tikzcd}
$$

where $I \cap J$ corresponds to the abstract union of "measurable set",
$I+J$ corresponds to the abstract intersection of "measurable set". So the direct limit and inverse limit of $\Pi_{I}(-)$ and $\pi_{I}(-)$ could be defined.
Let $\Omega$ be a set of some open ideal of $A$ such that:
$$I,J \in \Omega \implies I\cap J,I+J\  \textup{in} \ \Omega.$$ 
There exists a canonical morphism:
$$
\begin{tikzcd}
&\underset{I\in \Omega}{\varinjlim}\Pi_{I}(M) \arrow[r] &M \arrow[r] & \underset{I\in \Omega}{\varprojlim}\pi_{I}(M). 
\end{tikzcd}
$$
\end{prop}
\begin{rem}
	In general, $\underset{I\in \Omega}{\varinjlim}\Pi_{I}(M)$ can be viewed as "\textbf{categorical interior} of $M$", $\underset{I\in \Omega}{\varprojlim}\pi_{I}(M)$ can be viewed as "\textbf{categorical closure} of $M$".
	
	The phenomenon is similar as the categorical definition of interior and closure of topological space $X$:
$$
  \begin{tikzcd}
		&\underset{U \ \textup{is open}}{\varinjlim}U=\underset{U\subset X}{\bigcup}U \arrow[r] &X \arrow[r] & \underset{X\subset V}{\bigcap}V=\underset{V \ \textup{is closed}}{\varprojlim}V. 
  \end{tikzcd}
$$
	In the case of Hilbert space, the orthogonal projection theorem tells us the closed subspace of Hilbert space is both "categorical open" and "categorical closed". So the category of Hilbert space admits excellent measure-theoretic structure.
\end{rem}
\begin{prop}
In the functor category $\Fct(\AB,\AB)$ and $\Fct(\AB_S,\AB_S)$ we have:
    \begin{align*}
	&0.\Pi_{(1)}=0_{\AB},\pi_{(1)}=0_{\AB_S};\\
	&1.\Pi_{(0)}=Id_{\AB},\pi_{(0)}=Id_{\AB_S};\\
	&2.\Pi_{I}^2=\Pi_{I},\pi_{I}^2=\pi_{I}; \\
	&3.\Pi_{I}\circ \Pi_{J}=\Pi_{I+J},\pi_{I}\circ \pi_{J}=\pi_{I+J}; \\
	&4.\Pi_{I}\circ \Pi_{J}=0 \iff I+J=(1);\\
	&5.\pi_{I}\circ \pi_{J}=0 \iff I+J=(1).
	\end{align*}
\end{prop}
\begin{prop}
Let $\Comp A=\left \{ I \subset A \;\middle|\; I \textup{ is open ideal} \right \}$ be the set of open ideal of $A$, Let $\Omega$ be any index set. we have:
   \begin{align*}
	&1.I,J\in \Comp A\Rightarrow I\cap J \in \Comp A;\\
	&2.\left \{ I_k \right \}_{k \in \Omega}\subset \Comp A \Rightarrow \sum_{k \in \Omega} I_k \in \Comp A.
   \end{align*}
$\Comp A$ has a lattice structure with $\cap$ and $\sum$.
\end{prop}

Let $R$ be a object in the category $\OA$, not necessary commutative. Every morphism: $f\in \Hom_{\OA}\left(A,R\right)$ induces a map:
\begin{align*}
	f^*:\Comp R &\to \Comp A\\
    V & \mapsto f^{-1}\left(V\right). 
\end{align*}
There exists a \textbf{relative topology} on $\Hom_{\OA}\left(A,R\right)$, which is defined by the weakest topology makes every element in $A$ be a continuous function on $\Hom_{\OA}\left(A,R\right)$. Moreover, there exists a categorical ismorphsim:\footnote{It comes from the definition of continuous morphism.}
$$
\Hom_{\OA}\left(A,R\right)\simeq \varprojlim_{\varepsilon^{*} \in \mathfrak{m}_p^{+}}\varinjlim_{\varepsilon \in \mathfrak{m}_p^{+}}\Hom\left(A/\varepsilon A,R/\varepsilon^{*} R\right).
$$

Now let's study the spectral measure theory on the formal group scheme $$
\G: =\left \{ f(T)=\sum_{n=-\infty}^{\infty}a_nT^n,a_n \in \mathcal{O}_p,a_n \to 0  \right \}= \mathcal{O}_p\left \langle t,t^{-1} \right \rangle  
$$
\begin{defn}
Let
$$
u(\G)=\left \{ f\in \mathcal{O}_p\left [ t,t^{-1}\right ] \;\middle|\; f(t)=c\prod_{\lambda\in \mathcal{O}_p^{\times}}(t-\lambda)^{n_{\lambda}},c \in \mathcal{O}_p^{\times},n_{\lambda} \in \Z \right \}.
$$
we call $u(\G)$ \textbf{unit polynomial} of $\G$. It's obvious to see that:
$$
f,g\in u(\G) \implies fg \in u(\G)
$$
\end{defn}
\begin{defn}
Let $\varepsilon  \in \mathfrak{m}_p^{+}, f \in u(\G)$, the \textbf{unit ideal} is defined as:
$$
I_{\varepsilon ,f}=\begin{cases}
	\left (\varepsilon,f\right )  & \ \ \varepsilon  \in \mathfrak{m}_p \\
	\bigcup_{\varepsilon \in \mathfrak{m}_p}I_{\varepsilon ,f}    & \ \  \varepsilon=1^{-}
\end{cases}
$$
\end{defn}
In geometry, $I_{\varepsilon ,f}$ represents some balls whose radius are $\epsilon$ in $\mathfrak{m}_p^{+}$. We want to use $I_{\varepsilon ,f}$ to parameterize the action of formal group scheme. Let $\Rt f$ donate the set of all root(including multiplicity) of $f$, $d(\Rt f,\Rt g)$ be the distance of set $\Rt f,\Rt g$, $res(f,g)$ be the resultant of $f,g$.
\begin{thm}
Let $p \ge 3$, the following conditions are equivalent:
	\begin{align*}
	&1.I_{\varepsilon,f}+I_{\varepsilon,g}=(1);\\
	&2.\Pi_{I_{\varepsilon,f}}\circ\Pi_{I_{\varepsilon,g}}=0;\\
	&3.\pi_{I_{\varepsilon,f}}\circ\pi_{I_{\varepsilon,g}}=0;\\
	&4.d(\Rt f,\Rt g)=1;\\
	&5.res(f,g)\in \mathcal{O}_p^{\times};\\
	&6.\G/I_{\varepsilon,fg}\simeq \G/I_{\varepsilon,f}\times \G/I_{\varepsilon,g}.
    \end{align*}
\end{thm}
\begin{proof}
	1 $\Leftrightarrow$ 2 $\Leftrightarrow$ 3: Noted that: $$\Pi_{I_{\varepsilon,f}}\circ\Pi_{I_{\varepsilon,g}}=\Pi_{I_{\varepsilon,f}+I_{\varepsilon,g}}$$
	$$\pi_{I_{\varepsilon,f}}\circ\pi_{I_{\varepsilon,g}}=\pi_{I_{\varepsilon,f}+I_{\varepsilon,g}}$$
	
	1 $\Leftrightarrow$ 4: Let $\tilde{f},\tilde{g}$ be the reduction of $f,g$ in the ring $\overline{\Fp}\left [t,t^{-1}\right ] $, then $\tilde{f},\tilde{g}$ are coprime. We have:
	$$
	\exists k(t),l(t) \in \overline{\Fp}\left [t,t^{-1}\right ],
	k(t)\tilde{f}(t)+l(t)\tilde{g}(t)=1.
	$$
    If $\lambda \in \overline{\Fp}$ is a common root of $\tilde{f},\tilde{g}$, then we have:
    $$
    k(\lambda)\tilde{f}(\lambda)+l(\lambda)\tilde{g}(\lambda)=0,
    $$
    which leads to a contradiction. So we have: $d(\Rt f,\Rt g)=1$. 
    
     Suppose $d(\Rt f,\Rt g)=1$, on the one hand, the reduction $\tilde{f},\tilde{g}$ are coprime since $\tilde{f},\tilde{g}$ has no common root. On the other hand, the reduction map:
 	$$
 	\mathcal{O}_p\left \langle t,t^{-1} \right \rangle \to \overline{\Fp}\left [t,t^{-1}\right ]
 	$$
 	is surjective. There exists a lifting of $k(t),l(t)$ such that:
 	$$I_{\varepsilon,f}+I_{\varepsilon,g}=(1).$$
 	
 	4 $\Leftrightarrow$ 5: From the definition of resultant.
 
 	5 $\Leftrightarrow$ 6: There exists $k(t),l(t) \in \mathcal{O}_p\left [ t,t^{-1} \right ] $ such that:
 	$$
 	k(t)f(t)+l(t)g(t)=res(f,g).
 	$$
 	We can define two orthogonal projection $P_1,P_2 \in \G/I_{\varepsilon,fg}$:
 	   \begin{align*}
 		&1.P_1(t)=k(t)f(t)/res(f,g),P_2=l(t)g(t)/res(f,g);\\
 		&2.P_1^2=P_1,P_2^2=P_2;\\
 		&3.P_1+P_2=1,P_1P_2=P_2P_1=0; \\
 		&4.P_1(t)g(t)=0,P_2(t)f(t)=0; \\
        &5.P_1(t)f(t)=f(t),P_2(t)g(t)=g(t);\\
        &6.\forall h(t) \in \G/I_{\varepsilon,fg},h(t)=P_1(t)h(t)+P_2(t)h(t).
 	    \end{align*}
From the argument of Chinese remainder theorem, there exists a isomorphism: 
$$
\G/I_{\varepsilon,fg}\overset{P_1\times P_2}{\simeq} \G/I_{\varepsilon,g}\times \G/I_{\varepsilon,f}.
$$
\end{proof}
\begin{rem}
	This theorem tells us how to define the orthogonality property on $\G$.
\end{rem}
$\mathcal{O}_p^{\times}$ can be parameterized by the set of Teichmüller element $T\left ( \mathcal{O}_p^{\times} \right )$, which is a lifting of $\overline{\Fp}$. The distance of different Teichmüller element is always 1.

 Moreover, there exists a one to one correspondence between the balls in $\mathcal{O}_p^{\times}$ whose radius is 1 and $T\left ( \mathcal{O}_p^{\times} \right )$.
$$
T\left ( \mathcal{O}_p^{\times} \right )=\left \{ x\in \mathcal{O}_p^{\times}  \;\middle|\;  \exists k\in \N,x^{p^{k}}=x\right \} 
$$
\begin{cor}
Suppose $f \in u(\G), \lambda \in T\left ( \mathcal{O}_p^{\times} \right )$. Let 
$$
\Rt_{\lambda}(f)=\left \{ t_{\lambda}\in \Rt(f)   \;\middle|\;   \  \left |t_{\lambda}-\lambda \right |_p<1   \right \}, f_{\lambda}(t)=\prod_{t_{\lambda}\in \Rt_{\lambda}(f)} (t-t_{\lambda}).
$$
We have:
$$f(t)=\prod_{\lambda\in T(\mathcal{O}_p)}f_{\lambda}(t).$$
There exists a decomposition:
$$
\G/I_{\varepsilon,f}\simeq \prod_{\lambda\in T(\mathcal{O}_p)}\G/I_{\varepsilon,f_{\lambda}}
$$
Let 
$$
u_{\lambda}(\G)=\left \{ f\in u(\G)\;\middle|\; \Rt(f)=\Rt_{\lambda}(f)\right \}.
$$
\end{cor}
\begin{defn}
	We call $f_{n}(t)=t^n-1,n\in \Z$ \textbf{principal unit polynomial}, $I_{\varepsilon,n}=I_{\varepsilon,f_{n}(t)}$ \textbf{principal unit ideal}.
\end{defn}
\begin{prop}
	Every unit ideal $I_{\varepsilon ,f}$ include a principal unit ideal $I_{\varepsilon ,n}$ for some $n \in \Z$.
\end{prop}
\begin{proof}
 Suppose the reduction of $f$ in $\overline{\Fp}\left [t,t^{-1}\right ] $ is $\tilde{f}$, since every invertible matrix in $M_k(\overline{\Fp})$ has finite order, there exists $N$ such that $\tilde{f}(t) \mid t^N-1$. In the $\G/I_{\varepsilon ,f}$ we have:
 $$
\left |t^N-1\right |< 1.
 $$
 So there exists $l$ such that:
 $$
 \left |t^{p^lN}-1\right |<  \left |\varepsilon \right |\implies I_{\varepsilon ,f} \supseteq  I_{\varepsilon ,p^lN}.
 $$
\end{proof}
\begin{rem}
$\G/I_{\varepsilon,n}$ is a subgroup scheme of $\G$. However, $\G/I_{\varepsilon,n}$ behaves bad in the category of non-commutative ultrametric $\mathcal{O}_p$-Banach algebra.
\end{rem}
\begin{prop}
	   \begin{align*}
		&1.I_{\varepsilon,n}+I_{\varepsilon,m}=I_{\varepsilon,gcd(n,m)};\\
		&2.I_{\varepsilon,n}\cap I_{\varepsilon,m}=I_{\varepsilon,lcm(n,m)}.
	\end{align*}
\end{prop}
Let $M \in Obj(\GB_S)$, define $M_\varepsilon=M/\varepsilon M$. We have:
\begin{prop}
	\begin{align*}
	\underset{n \ge 1}{\varinjlim}\Pi_{I_{\varepsilon,n}}(M_{\varepsilon})&=	\underset{n \ge 1}{\varinjlim}\Hom_{\GB}(\G/I_{\varepsilon,n},M_{\varepsilon})\\
    &=\underset{n \ge 1}{\varinjlim}\Hom_{\GB}(\G/I_{\varepsilon,n!},M_{\varepsilon})\\
    &=\underset{n \ge 1}{\varinjlim}\underset{\lambda\in T\left ( \mathcal{O}_p^{\times} \right)}{\oplus}M_{\varepsilon,n,\lambda,tor}\\
    &=\underset{\lambda\in T\left ( \mathcal{O}_p^{\times} \right )}{\oplus}\underset{n \ge 1}{\varinjlim}M_{\varepsilon ,n,\lambda,tor}\\
    &=\underset{\lambda\in T\left ( \mathcal{O}_p^{\times} \right )}{\oplus}M_{\varepsilon ,\lambda,tor}\\
    &=M_{\varepsilon ,tor},
	\end{align*}
	we have:
	\begin{align*}
	M_{\varepsilon ,tor}&=\left \{ m\in M_{\varepsilon} \;\middle|\; \exists f\in u(\G),f(t)m=0\right \} \\
	M_{\varepsilon ,\lambda,tor}&=\left \{ m\in M_{\varepsilon,tor} \;\middle|\; \exists f\in u_{\lambda}(\G),f(t)m=0\right \}.
	\end{align*}
\end{prop}
\begin{prop}
	\begin{align*}
		\underset{n \ge 1}{\varprojlim}\pi_{I_{\varepsilon,n}}(M_{\varepsilon})&= \underset{n \ge 1}{\varprojlim}M_{\varepsilon}/(t^n-1)M_{\varepsilon}\\
		&=\underset{n \ge 1}{\varprojlim}M_{\varepsilon}/(t^{n!}-1)M_{\varepsilon}\\
        &=\underset{n \ge 1}{\varprojlim}\prod_{\lambda \in T\left ( \mathcal{O}_p^{\times} \right )} M_{\varepsilon,n,\lambda} \\
        &=\prod_{\lambda \in T\left ( \mathcal{O}_p^{\times} \right )}\underset{n \ge 1}{\varprojlim} M_{\varepsilon,n,\lambda} \\
        &=\prod_{\lambda \in T\left ( \mathcal{O}_p^{\times} \right )} M_{\varepsilon,\lambda}\\
        &=\hat{M_{\varepsilon}} .
	\end{align*}
\end{prop}
\begin{rem}
	$M_{\varepsilon ,tor}$ and $\hat{M_{\varepsilon}}$ could be $0$ for some $M$. Let $Frac\left ( \G \right )$ be the fraction field of $\G$, which has a multiplicative norm induced by $\left |-\right |_{\G}$. Let $K$ be the completion of $Frac\left ( \G \right )$, $R$ be the integral ring of $K$, which is a strict $\GB$ module. We have: 
	$$
	R_{\varepsilon ,tor}=\hat{R_{\varepsilon}}=0.
	$$
	This example shows that the generic point of $\G$ has an influence on the category $\GB$.
\end{rem}
\begin{prop}
	Let $M \in Obj(\GB_S)$, $M_i=\cup_{\varepsilon \in \mathfrak{m}_p}\varepsilon M$, $M_1=M/M_i$. Suppose: 
	$$
	M_{\varepsilon ,tor}=\hat{M_{\varepsilon}}=0,\forall \varepsilon \in \mathfrak{m}_p^{+},
	$$ 
	then $M_1$ is a $\overline{\Fp}\left ( t \right ) $-linear space.
\end{prop}
\begin{proof}
	Suppose $M_1\ne 0$, for any $\lambda \in \overline{\Fp}$, there exists a exact sequence:
	$$
	\begin{tikzcd}
	&0\arrow[r] &\ker\left(t-\lambda\right)\arrow[r] &M_1 \arrow[r,"t-\lambda"] &M_1\arrow[r]  &\coker\left(t-\lambda\right)\arrow[r]  &0.
	\end{tikzcd}
	$$
	We have: 
	$$
	\ker\left(t-\lambda\right)=\coker\left(t-\lambda\right)=0.
	$$
	So $t-\lambda$ is invertible. Hence $M_1$ is a $\overline{\Fp}\left ( t \right ) $-linear space.
\end{proof}
\begin{defn}
    We call $M_{\infty}=\overline{\Fp}\left ( t \right )\underset{\G}{\otimes}M$ \textbf{infinite place of spectrum decomposition}. Let $I$ be a open ideal of $\G$. We call $\Pi_{I}(M_{\varepsilon}),\pi_{I}(M_{\varepsilon})$ \textbf{compact place of spectrum decomposition}. We call set:
    $$
    \Sp(M)=\left \{ M_{\infty};M_{\varepsilon ,\lambda,tor};M_{\varepsilon,\lambda},\forall \varepsilon\in \mathfrak{m}_p^{+},\lambda \in T\left ( \mathcal{O}_p^{\times} \right )\right \}  
    $$
    \textbf{the spectrum decomposition of $M$}. \footnote{Noted that the spectrum decomposition of $M$ has a correspondence to $\Berk \G$, which is the spectrum of $\G$ in Berkovich sense.}
\end{defn}
\begin{thm}(completeness of spectrum)
	Suppose there exist $\psi \in M$ such that\footnote{$\psi$ can be understood as a wave function.} $\left |\psi\right |=1 $, then $\Sp(M)$ is not $\left \{0\right \}$.
\end{thm}
\begin{proof}
	We have:
	$$
	M_1 \ne 0.
	$$
	Suppose $\Sp(M)$ is $\left \{0\right \}$, then $M_1$ is a $\overline{\Fp}\left ( t \right ) $-linear space, which leads to a contradiction.
\end{proof}

\section{The Galois theory of operators}
The topological property of $p$-adic unitary operator is related to $\Q_p^{\times}$. More precisely, the abelian extension theory of $\Q_p$. In local class field theory, the pro-finite completion of $\Q_p^{\times}$ corresponds to $\Gal\left (\Q_p^{ab}\mid \Q_p\right )$. There exists a local \textbf{Artin morphism} 
$$\theta: \Q_p^{\times}\to \Gal\left (\Q_p^{ab}\mid \Q_p\right )$$
such that the following diagram commutes:
$$
\begin{tikzcd}
0 \arrow[r] & \Z_p^{\times} \arrow[r]\arrow[d,"\simeq"] & \Q_p^{\times}\arrow[r]\arrow[d,hookrightarrow,"\theta"] & \Z \arrow[r]\arrow[d,hookrightarrow] &0\\
0 \arrow[r] & \Gal\left (\Q_p^{ab}\mid \Q_p^{ur}\right ) \arrow[r] & \Gal\left (\Q_p^{ab}\mid \Q_p\right ) \arrow[r] & \Gal\left (\Q_p^{ur}\mid \Q_p\right ) \arrow[r] & 0
\end{tikzcd}
$$
where $\theta$ is an almost isomorphism.
In this section, we are about to talk about the Galois theory of operators.
\begin{prop}
	We define the following group:
	\begin{align*}
	B\left ( 1 \right )&=\left \{ x\in \mathcal{O}_{p}^{\times} \;\middle|\;\  \left | x-1 \right |_p <1\right \} \\
    T\left ( \mathcal{O}_p^{\times} \right )&=\left \{ x\in \mathcal{O}_p^{\times}  \;\middle|\; \exists k\in \N,x^{p^{k}}=x\right \}.
    \end{align*}
    Let $\sigma$ be the Frobenius map:
	\begin{align*}
	\sigma:\mathcal{O}_{p}^{\times}& \to \mathcal{O}_{p}^{\times} \\
    x& \mapsto x^p.
	\end{align*}
	Let $x\in \mathcal{O}_{p}^{\times}$, we have:
	\begin{align*}
	\lim_{n\to \infty}\sigma^{n!}(x) \to 1 & \Leftrightarrow x \in B\left ( 1 \right )\\
	\lim_{n\to \infty}\sigma^{n!}(x) \to x & \Leftrightarrow x \in T\left ( \mathcal{O}_p^{\times} \right ).
    \end{align*}
    Let $x\in \CC_{p}$, we have:
    \begin{align*}
    \lim_{n\to \infty}\sigma^{n!}(x) \textup{\ converges} &\Leftrightarrow x\in \mathcal{O}_{p}\\
    \lim_{n\to \infty}x^{n!} \to 1 &\Leftrightarrow x\in \mathcal{O}_{p}^{\times}\\.
    \end{align*}
    Finally, we have:
    \begin{align*}
    \mathcal{O}_{p}^{\times}&\simeq B\left ( 1 \right ) \times T\left ( \mathcal{O}_p^{\times} \right )\\
    x &\mapsto \left (\frac{x}{\underset{n\to \infty}{\lim}\sigma^{n!}(x) },\lim_{n\to \infty}\sigma^{n!}(x)\right ).
    \end{align*}
\end{prop}
\begin{defn}
Let $M\in Obj(\GB)$, suppose there exists $x\in M,\left |x\right |=1$. Let $U$ be a $p$-adic unitary operator on $M$, $\sigma:U\mapsto U^p$. We say $U$ is \textbf{continuous type} if the strong operator limit of $\left \{ \sigma^{n!}(U)\right \}_{n\in \N}$ converges to identity operator:
$$
s-\lim_{n\to \infty}\sigma^{n!}(U) \to 1.
$$
We say $U$ is \textbf{Teichmüller type} if the strong operator limit of $\left \{ \sigma^{n!}(U)\right \}_{n\in \N}$ converges to $U$.
$$
s-\lim_{n\to \infty}\sigma^{n!}(U) \to U.
$$
We say $U$ is \textbf{pro-finite type} if the strong operator limit of $\left \{ U^{n!}\right \}_{n\in \N}$ converges to identity operator:
$$
s-\lim_{n\to \infty}U^{n!} \to 1.
$$
Moreover, $U$ is \textbf{pro-finite type} if and only if:
$$M_{\varepsilon,tor}=M_{\varepsilon}$$
\end{defn}
The following theorem follows from \cite{NAUO}. We do not need the normal property of $p$-adic unitary operators.
\begin{thm}($p$-adic Stone's theorem)
$U$ is a $p$-adic unitary operator of \textbf{continuous type} if and only if the one-parameter unitary group $\left \{ U^{t}\right \}_{t\in \Z_p}$ can be well defined.

Moreover, there exists a $\Z_p^{\times}\simeq \Gal\left (\Q_p^{ab}\mid \Q_p^{ur}\right ) $ action on the one-parameter unitary group $\left \{ U^{t}\right \}_{t\in \Z_p}$, where $\Q_p^{ab}$ is the maximum abelian field extension of $\Q_p$, $\Q_p^{ur}$ is the maximum unramified extension of $\Q_p$.
\end{thm}
\begin{proof}
The definition of continuous type $p$-adic unitary operator is equivalent to the $p$-adic continuous property of group $\left \{ U^{n}\right \}_{n\in \Z}$ at $n=0$. So $\left \{ U^{n}\right \}_{n\in \Z}$ can be uniquely extended to $\left \{ U^{t}\right \}_{t\in \Z_p}$ with the following diagram commutes:
$$
\begin{tikzcd}
\Z \arrow[r] \arrow[d]& \left \{ U^{n}\right \}_{n\in \Z} \arrow[d]\\
\Z_p \arrow[r] & \left \{ U^{t}\right \}_{t\in \Z_p}
\end{tikzcd}
$$
Let $\left \{ U^{t}\right \}_{t\in \Z_p}$ be any one-parameter unitary group, $\alpha \in \Z_p^{\times}$, we can define:
\begin{align*}
    \phi_\alpha :\left \{ U^{t}\right \}_{t\in \Z_p} &\overset{\simeq}{\longrightarrow} \left \{ U^{t}\right \}_{t\in \Z_p}\\
	U &\mapsto U^{\alpha}.
\end{align*}
 $\left \{\phi_\alpha,\alpha \in \Z_p^{\times} \right \}$ is the ismorphism group of $\left \{ U^{t}\right \}_{t\in \Z_p}$.
 
 In the case of usual one-parameter unitary group on Hilbert space, the action of $\Z_p^{\times}$ corresponds to the skew-symmetry property of generator. The Hermite conjugate gives a $\Z/2\Z \simeq \Gal\left (\CC \mid \R \right )$-action.
\end{proof}
\begin{thm}(Spectral decomposition theorem of Teichmüller element)
	$U$ is a $p$-adic unitary operator of \textbf{Teichmüller type} if and only if there exists a spectral decomposition:
	$$
	\sum_{\lambda\in T\left ( \mathcal{O}_p^{\times} \right )}\pi_\lambda=1,\ \ \sum_{\lambda\in T\left ( \mathcal{O}_p^{\times} \right )}\lambda\pi_\lambda=U,\ \ \pi_\lambda^2=\pi_\lambda,\ \ \pi_\lambda\pi_{\lambda^*}=0(\forall \lambda \ne \lambda^*)
	$$
	the sum converges in strong operator topology.
	
	Moreover, there exists a $\hat{\Z}\simeq \Gal\left (\Q_p^{ur} \mid \Q_p \right ) $ action on $U$, where $\Q_p^{ur}$ is the maximum unramified  field extension of $\Q_p$ by joining the $T\left ( \mathcal{O}_p^{\times} \right )$.
\end{thm}
\begin{proof}
	We refer to \cite{zth} for the proof of spectral decomposition.
	We only show how to construct a $\Gal\left (\Q_p^{ur} \mid \Q_p \right ) $-action on $\left \{ U^{n}\right \}_{n\in \Z}$.
	
	First, $\Gal\left (\Q_p^{ur} \mid \Q_p \right ) $ has a pro-finite group action on $T\left ( \mathcal{O}_p^{\times} \right )$. Let $g\in \Gal\left (\Q_p^{ur} \mid \Q_p \right )$, we have a group ismorphism:
	\begin{align*}
	\psi_g: T\left ( \mathcal{O}_p^{\times} \right ) &\to T\left ( \mathcal{O}_p^{\times} \right )\\
		\lambda &\mapsto g(\lambda).
	\end{align*}
	Finally, $\Gal\left (\Q_p^{ur} \mid \Q_p \right )$ has a  group action on $p$-adic unitary operator of \textbf{Teichmüller type}. The construction is:
	\begin{align*}
		U&=\sum_{\lambda\in T\left ( \mathcal{O}_p^{\times} \right )}\lambda\pi_\lambda\\
		\psi_g(U)&=\sum_{\lambda\in T\left ( \mathcal{O}_p^{\times} \right )}g(\lambda)\pi_\lambda.
	\end{align*}
    Since the $\Gal\left (\Q_p^{ur} \mid \Q_p \right )$-group action on $T\left ( \mathcal{O}_p^{\times} \right )$ is pro-finite, the sum converges in strong operator topology.
\end{proof}
\begin{thm}(Jordan decomposition theorem of pro-finite unitary operator)
	$U$ is a $p$-adic unitary operator of \textbf{pro-finite type} if and only if there exists a Jordan decomposition:
	$$
    U=U_s U_n,
	$$
	where $U_s$ is a $p$-adic unitary operator of \textbf{Teichmüller type}, $U_n$ is a $p$-adic unitary operator of \textbf{continuous type}. The Jordan decomposition of $U$ is unique.
	
	Moreover, there exists a $\hat{\Z}\times \Z_p^{\times} \simeq \Gal\left (\Q_p^{ab} \mid \Q_p \right ) $ action on $U$, where $\Q_p^{ab}$ is the maximum abelian field extension of $\Q_p$.
\end{thm}
\begin{proof}
 Recall that $\mathcal{O}_{p}^{\times}\simeq B\left ( 1 \right ) \times T\left ( \mathcal{O}_p^{\times} \right )$. For any $x \in {(\mathcal{O}_{p}/\varepsilon)}^{\times}$, $x$ has a finite order:
 $$
 \exists n(x)\in \N, x^{n(x)}=1
 $$
 The Jordan decomposition of $x$ is uniquely defined by the limit:
 $$
 x_s=\lim_{n\to \infty}\sigma^{n!}(x),\ x_n=\frac{x}{x_s}
 $$
 From the definition of pro-finite type $p$-adic unitary operator, for any $m$ in $M$, we have:
 $$
 \lim_{k\to \infty}U^{k!}m= m.
 $$
 Let $\varepsilon \in \mathcal{O}_p-\left \{0\right \}$, $m_{\varepsilon} \in M_{\varepsilon}$ be the reduction of $m$, we have:
 $$
 \exists n \in \N,\ U_{\varepsilon}^{n}m_{\varepsilon}= m_{\varepsilon}.
 $$
 The Jordan decomposition of $U_{\varepsilon}$ is defined by:
 $$
 U_{\varepsilon,s}\left ( m_{\varepsilon} \right ) =\lim_{n\to \infty}\sigma^{n!}(U_{\varepsilon})\left ( m_{\varepsilon} \right ) ,\ U_{\varepsilon,n}=\frac{U_{\varepsilon}}{U_{\varepsilon,s}}.
 $$
 
Since we have: 
$$
M= \varprojlim_{\varepsilon\in \mathfrak{m}_p^{+}} M_{\varepsilon}
$$

The limit $\lim_{n\to \infty}\sigma^{n!}(U)\left ( m\right )$ converges. Let's define:
$$
    U_{s}\left ( m\right )=\lim_{n\to \infty}\sigma^{n!}(U)\left ( m\right ),\ U_{n}=\frac{U}{U_{s}}
$$

The Galois group $\Gal\left (\Q_p^{ab} \mid \Q_p \right )$ action on $U$ is defined by the product:
\begin{align*}
\Gal\left (\Q_p^{ab} \mid \Q_p \right )&\simeq \hat{\Z}\times \Z_p\\
\left (g,h\right )(U)&=g(U_s)h(U_n).
\end{align*}
\end{proof}
Let $\sigma \in \Gal\left(\CC_p\mid \Q_p\right)$. The categorical explanation is that the Galois group acts on the set of open ideal of $\G$, which is induced by the following diagram:
$$
\begin{tikzcd}
	\G \arrow[r,"\sigma"]& \G \\
	\sum_{k\in \Z}a_k t^k  \arrow[r,mapsto,"\sigma"] & \sum_{k\in \Z}\sigma\left(a_k\right)t^k \\
	\Comp\G  & \Comp\G \arrow[l,swap,"\sigma^{-1}"].
\end{tikzcd}
$$

Moreover, every isomorphism $f$ of $\G$ induce a bijective map: 
$$
f^*: \Comp \G \rightarrow \Comp \G
$$
The involution $i: t\to t^{-1}$ and the rotation $r: t \to \alpha t,\left | \alpha  \right |=1$ can be realized as the geometric symmetry on $\Comp \G$.

In fact, the definition of Hermite/unitary operator highly depends on the structure of field. The definition of inner product relies on the positive property of square in $\R$, which is important in quantum theory. The Galois group can be understood as a symmetry of operators.

\begin{rem}
	
	There exists a measure theory on any pro-finite group $G$. For any open-closed normal subgroup $N$ of $G$, define the volume of $N$ by:
$$
	m\left (N \right )= \frac{1}{\left |G/N \right | }.
$$

    Let $L$ be a Galois extension of field $K$, then $\Gal\left ( L\mid K \right )$ has a measure theory for any finite sub Galois extension $M$ of $K$. We can define the volume of $\Gal\left ( L\mid M \right )$ by:
    $$
    m\left (\Gal\left ( L\mid M \right )\right )= \frac{1}{\left |M/K \right | }=\frac{1}{\left |\Gal\left ( M\mid K\right ) \right |}.
    $$
    Let $c+d\Z \subset \Z,d \in \Z-\left \{0\right \}  ,c\in \Z$, define the volume of $c+d\Z$ by:
    $$
    m\left (c+d\Z \right )=\frac{1}{\left |d\right | }.
    $$
    The volume of $c+d\Z$ can be viewed as the measure theory on pro-finite group $\hat{\Z}$.
\end{rem}
\section{Examples}
\begin{ex}(The spectral theory of $\G$ act on itself)
	Let $\G$ act on itself. We have: 
	$$
	\G=\mathcal{O}_p\left<t,t^{-1}\right>\simeq c_{0}\left(\Z\right),
	$$
	where $t$ is a shift operator on $c_{0}\left(\Z\right)$. Let $c+d\Z$ be a arithmetic sequence. We can define the sum $S_{c+d\Z}$, which is a "integral of $c+d\Z$":
	$$
	f\left( t\right) =\sum_{n=-\infty}^{\infty}a_nt^n,
	S_{c+d\Z}\left( f\left( t\right) \right)= \sum_{n=-\infty}^{\infty}a_{c+dn}.
	$$
	We have:
	$$
	S_{c+d\Z}\left(f\left( t\right)\right) =S_{c+d\Z}\left(t^df\left( t\right) \right).
	$$
	Moreover:
	$$
	S_{c+d\Z}=\sum_{c^{*}=1}^{d^{*}}S_{c+c^{*}d+dd^{*}\Z}
	$$
	which is the addictive property of integral. Finally, it is equivalent to use the right projection functor:
	$$
	\G/\left(t^d-1\right) \otimes \mathcal{O}_p\left<t,t^{-1}\right>\simeq \mathcal{O}_p\left[t\right]/\left(t^d-1\right)\simeq \mathcal{O}_p^{d},
	$$
	where the coefficient corresponds to $S_{c+d\Z}$. We have:
	$$
	f\left(t\right) =0 \Leftrightarrow S_{c+d\Z}\left(f\left(t\right) \right) =0, \forall c\in \Z,d \in \Z-\left \{0\right \} .
	$$
\end{ex}
\begin{ex}(Spectrum shift operator)
	Let $X$ be a $p$-adic Hermite operator act on $\CC_p$-ultrametric Banach space $V$. Suppose the spectrum of $X$ lies in $\Z_p$. We have:
	$$
	X=\int_{\Z_p}\lambda dE_{\lambda }.
	$$
	Let $U$ be a $p$-adic unitary operator on $X$. We say $U$ is a \textbf{spectrum shift operator} if the following condition holds:
	$$
	UX-XU=U.
	$$
	We have:
	\begin{align*}
		UXU^{-1}=X+1&\implies \int_{\Z_p}\lambda UdE_{\lambda }U^{-1}=\int_{\Z_p}\lambda+1 dE_{\lambda }=\int_{\Z_p}\lambda dE_{\lambda-1}\\
		&\implies UdE_{\lambda }U^{-1}=dE_{\lambda-1}.
	\end{align*}
	The conjugate of $U$ make the spectral measure of $X$ shift by 1. Suppose $U$ is a $p$-adic unitary operator of \textbf{continuous type}, then $X$ has a continuous spectrum $\Z_p$ shift by $\left \{U^{t}\right \}_{t\in \Z_p}$. Let $C\left(\Z_p,\CC_p\right)$ be the continuous function of $\Z_p$ valued in $\CC_p$. Suppose:
	$$
	U\left(f\left(x\right) \right) =f\left(x+1\right) ,X\left(f\left(x\right) \right)=xf\left( x\right).
	$$
	We have: $U$ is a spectrum shift operator of $X$. Let:
	\begin{align*}
	 a^{+}=XU^{-1},&\ a^{-}=U-I,\ H=a^{+}a^{-};\\
	&a^{-}a^{+}-a^{-}a^{+}=1.
    \end{align*}
    Then $\left( a^{+},a^{-}\right) $ is the creation and annihilation operators of $H$.
    A non-commutative torus $T_{\xi}$ can be described by the following commutation relation:
    $$
    UV=\xi VU,
    $$
    where $U,V$ is $p$-adic unitary operator, $\xi \in \mathcal{O}_p^{\times}$ is a number.  Let $T_{\xi}$ be the $\mathcal{O}_p$-ultrametric algebra generated by $U,V$. The conjugation of V makes the spectrum of $U$ shift by $\xi$. Suppose $\left|\xi-1\right|< 1,\xi \ne 0$, $T_{\xi}$ has a reduction $T_{\varepsilon,\xi}$ such that $T_{\varepsilon,\xi}$ is commutative even though $T_{\xi}$ is non-commutative. Let $n,m \in \Z$, we have:
    $$
    U^{n}V^{m}=\xi^{nm}V^{m}U^{n}.
    $$
    We can find $n,m$ such that $\xi^{nm} \to 1$. The algebra $T_{\xi^{nm}}$ tends to be a commutative algebra.
\end{ex}
\begin{ex}($p$-adic unitary matrix group)
	Let $\GL_n\left( \Z_p\right)$ be the group of n*n invertible matrix over $\Z_p$, then $\forall U \in \GL_n\left( \Z_p\right)$, $U$ is pro-finite type $p$-adic unitary operator.
	\begin{proof}
		The reduction of $U$ in $\GL_n\left(\Fp\right) $ has a finite order $k$. So we have: 
		$$
		U^{k}\in I+pM_n\left(\Z_p\right)\implies U^{n!}\to I.
		$$
	\end{proof}
	Hence we have the Jordan decomposition of $U$:
	$$
	U=U_sU_n,
	$$
	where $U_s$ is a $p$-adic unitary matrix of \textbf{Teichmüller type}, $U_n$ is a $p$-adic unitary matrix of \textbf{continuous type}. 
	
	Let $\Z_p^n$ be the unit ball in $\Q_p^n$, the nature action of $\GL_n\left(\Z_p\right) $ on $\Z_p^n$ induce a group action on the continuous function on $\Z_p^n$. Let $g \in \GL_n\left(\Z_p\right)$,  $g^*$ be the action on continuous function. It is easy to check that $g^*$ is a $p$-adic unitary operator of \textbf{pro-finite type}.
	
	There exists a more accurately decomposition for matrix in $\GL_n\left(\Fp \right)$. For a arbitrary $k \in \N$, Let $\mathbb{F}_{p^k}$ be the finite field which has $p^k$ elements, $\mathbb{F}_{p^k}^{*}$ act on itself by left multiply, which is $\Fp$-linear. Let $t_k$ be the generator of $\mathbb{F}_{p^k}^{*}$, which can be realized as a $k*k$ invertible matrix $t_k$ in $\GL_k\left(\Fp\right)$ of Teichmüller type:
	$$
	\sigma^{k}\left(t_k\right)=t_k^{p^k}=t_k
	$$
	This is not canonical but it would be useful. Let 
	\begin{align*}
	\iota_k:\GL_k\left(\Fp\right) &\to \GL_n\left(\Fp\right)\\
	U &\mapsto \begin{pmatrix}
		U  & \\
		& I_{n-k}
	\end{pmatrix}
	\end{align*}
	be the embedding map, $T_k=\iota_k\left(t_k\right)$, define the set:
	\begin{align*}
	\Phi&=\left \{ T\in \GL_n\left (\Fp\right ) \  \middle| \ T=T_n^{m_n}...T_{2}^{m_2}T_{1}^{m_1};1\le m_k\le p^k-1,1\le k\le n \right \} \\
    B&=\left \{N \in \GL_n\left (\Fp\right ) \  \middle| \ N=\begin{pmatrix}
	1  & n_{12} & n_{13} & ...\\
	0 & 1 & n_{23} & ...\\
	0 & 0 & 1 & ...\\
	... & ... & ... & 1
\end{pmatrix};n_{ij}\in \Fp,1\le i,j \le n\right \}.
	\end{align*}
We have:
	  \begin{align*}
   &\left | \Phi \right |=\prod_{k=1}^{n}\left (p^k-1\right ),
   \left | B \right |=p^{\frac{n(n-1) }{2}}\\
     &\left | \Phi \right | \left | B\right |= \left | \GL_n\left(\Fp\right)\right |= \prod_{k=1}^{n}\left (p^n-p^k\right )
      \end{align*}
\begin{prop}
	Let $U\in \GL_n\left(\Fp\right)$, there exists a decomposition:
	$$
	U=TN;\ T\in\Phi,N \in B.
	$$
	Moreover, the pair of matrices $\left( T,N\right) $ is unique.
\end{prop}
\begin{proof}
	First, the case $n=1$ is obvious. Suppose the proposition holds for $n=k,k\in \N$, we show that the proposition holds for $n=k+1$. Let's define the affine transformation subgroup in $\GL_{k+1}\left(\Fp\right)$:
	$$
	\Aff_k=\left \{ g\in\GL_{k+1}\left (\Fp\right ) \  \middle| \ g=\begin{pmatrix}
		a_{11}  & ... & a_{1k} & b_1\\
		... & ... & ... & ...\\
		a_{k1}  & ... & a_{kk} & b_k\\
		0  & ... & 0 & 1
	\end{pmatrix}   \right \}
	$$
	It can be showed that $\Aff_k$ is generated by the following matrix:
	\begin{align*}
	&T=T_k^{m_k}...T_{2}^{m_2}T_{1}^{m_1};1\le m_l\le p^l-1,1\le l\le k \\
	&N=\begin{pmatrix}
		1  & ... & 0 & b_1\\
		... & ... & ... & ...\\
		0  & ... & 1 & b_k\\
		0  & ... & 0 & 1
	\end{pmatrix} ; b_l \in \Fp ,1\le l\le k
	\end{align*}
    Finally, the eigenvalue of elements in $\Aff_k$ lies in $\mathbb{F}_{p^k}$, the eigenvalue of $T_{k+1}$ lies in $\mathbb{F}_{p^{k+1}}$. So we have:
    \begin{align*}
    	&\left \{ T_{k+1}^{m_{}} \right \}_{1\le m\le p^{k+1}-1}\cap \Aff_k=\left \{I_{k+1}\right \}\\
    	&\left | \Aff_k \right |=\left | \GL_k\left ( \Fp \right )  \right |p^k \\
    	&\left | \Aff_k \right |\left( p^{k+1}-1\right) =\left | \GL_{k+1}\left ( \Fp \right )  \right |
    \end{align*}
\end{proof}
Suppose $p\ge 3$, there exists a Teichmüller lift for the set $\Phi$ to the group $\GL_n\left( \Z_p\right) $. Let $\Upsilon \subset \GL_n\left( \Z_p\right)$ be a lift set of $\Phi$. Let $\mathcal{B}$ be the group:
 $$\mathcal{B}=\left \{N \in \GL_n\left (\Z_p\right ) \  \middle| \ 
 N\in \begin{pmatrix}
 	1+p\Z_p & \Z_p & ...\\
 	p\Z_p& 1+p\Z_p & ... \\
 	...& ...& ... 
 \end{pmatrix}
 \right \}.
$$
 \begin{thm}
    	Let $U\in\GL_n\left(\Z_p\right) $, there exists a decomposition:
    	$$
    	U=TN;\  T\in \Upsilon, N\in \mathcal{B}.
    	$$
    	Moreover, the pair of matrices $\left( T,N\right) $ is unique.
 \end{thm}
\end{ex}
\begin{ex}($p$-adic time evolution)
	Considering a $p$-adic differential equation which is defined over a $\mathcal{O}_p$-ultrametric Banach space $X$. Let $H$ be a linear operator on $X$, $\psi:\Z_p \to X$ be a continuous function, \textbf{the $p$-adic time evolution equation} is defined as:
	$$
	\frac{\mathrm{d} \psi}{\mathrm{d} t} =H\psi.
	$$
    The $p$-adic time evolution equation has a formal solution:
    $$
    \psi\left(t\right) =e^{Ht}\psi(0), \ t \in p\Z_p.
    $$
    The formal solution sometimes can not be extended continuously to $\Z_p$. Noted that $e^{Ht}$ is a $p$-adic unitary operator of continuous type. We should find a more reasonable definition.
    \begin{defn}
    Let $X$ be a $\mathcal{O}_p$-ultrametric Banach space. Let $H$ be a linear operator on $X$, $U$ be a unitary operator on $X$. Suppose we have: $HU=UH$. Let $\psi_k:\Z_p \to X,k\in \Z$ be a family of continuous function, we say the following equation is \textbf{the $p$-adic time evolution equation}:
    $$
   \begin{cases}
   	\frac{\mathrm{d} \psi_k}{\mathrm{d} t} =H\psi_k,\ \  \forall t\in \Z_p \\
   	\psi_{k+1}=U\psi_k, \ \  \forall k\in \Z
   \end{cases}
    $$
    \end{defn}
    \begin{prop}
    	$$\left | \psi_{k}\left ( t+\varepsilon  \right )  \right |_X= \left | \psi_{k}\left ( t\right )  \right |_X, \ \left |\varepsilon\right |_p<\frac{1}{p\left | H \right |}
    	,
    	\left | \psi_{k+1} \right |_X=\left |\psi_k\right |.
    	$$
    \end{prop}
\end{ex}
\section{Framework of $p$-adic quantum mechanics}
In this section, we want to establish a framework of $p$-adic quantum mechanics by spectral theory of $p$-adic operators. We refer to \cite{pMP} \cite{pQM} for previous work.

Let $X$ be a $\CC_p$-ultrametric Banach space, $M$ be the unit ball of $X$, $H$ be a $p$-adic Hermite operator on $X$, $U$ be a $p$-adic unitary operator on $X$, $\pi$ be a orthogonal projection operator on $X$, $\psi\in X, \left |\psi \right |_X=1$ be a $p$-adic wave function.

\begin{ax}($p$-adic probability interpretation)
The $p$-adic quantum system is described by a $p$-adic wave function $\psi$ in a $\CC_p$-ultrametric Banach space $X$. The event is represented by $\pi$. The probability is given by $\left |\pi\left (\psi\right )  \right |_X$. Suppose the event $\pi$ can be decomposed into a family of disjoint event $\left \{\pi_i\right \}_{i=1}^{k}$ satisfy the following condition:
$$
\pi=\sum_{i=1}^{k} \pi_{i},\ \ \pi_{i}\pi_{j}=\delta_{ij}\pi_i, \ \forall 1\le i,j \le k.
$$
Then the probability of $\pi$ is the supremum of $\left \{\pi_i\right \}_{i=1}^{k}$:
$$
\left |\pi\left(x\right)\right |=\sup_{1\le i\le k} \left |\pi_i\left(x\right)\right |.
$$
Moreover, the right projection functor with respect to $U$ has a similar probability interpretation.
\end{ax}
\begin{ax}($p$-adic time evolution)
The $p$-adic time evolution is described by the operator $\left(H,U\right)$, where $H$ is a $p$-adic linear operator (not necessary $p$-adic Hermite), $U$ is a $p$-adic unitary operator, $HU=UH$. $H$ represents a continuous evolution, $U$ represents a discrete evolution.
The $p$-adic time evolution equation is given by:
$$
\begin{cases}
	\frac{\mathrm{d} \psi_k}{\mathrm{d} t} =H\psi_k,\ \  \forall t\in \Z_p \\
	\psi_{k+1}=U\psi_k, \ \  \forall k\in \Z.
\end{cases}
$$
\end{ax}
\begin{ax}($p$-adic observable and $p$-adic spectrum)
Suppose $H$ is a $p$-adic Hermite operator of period $1$ with the symmetry of Galois group $\Gal\left(\Q_p^{ur}\mid \Q_p\right) $, there exists a orthogonal projection valued spectral integral:
$$
I=\int_{\Q_p}dE_{\lambda }  \ \ \   H=\int_{\Q_p}\lambda dE_{\lambda }.
$$

Suppose $U$ is a $p$-adic unitary operator, then $X$ has a projection functor valued spectral measure. Let $\psi \in M, \left|\psi\right|=1$ be a $p$-adic wave function, the spectrum decomposition:
$$
    \Sp(M)=\left \{ M_{\infty};M_{\varepsilon ,\lambda,tor};M_{\varepsilon,\lambda},\forall \varepsilon\in \mathfrak{m}_p^{+},\lambda \in T\left ( \mathcal{O}_p^{\times} \right )\right \}
$$
is complete.
\end{ax}

\begin{ax}($p$-adic quantum measurement)
The $p$-adic quantum measurement is a right projection functor $$\pi_{I}: \GB_S \to \GB_S,$$
where $M$ is viewed as $\GB_S$ module, $\psi \in M$. After the $p$-adic quantum measurement, the wave function $\psi$ restricts to $\pi_{I}\left(\psi\right)$.
\end{ax}

\section{Further discussion}
Here are some questions remain unsolved.
\begin{q}
	The behavior of open ideal seems to be a system which is both discrete and continuous. The reduction morphism makes $\G$ ignore the diameter which is smaller than $\varepsilon$. What is the $p$-adic $\zeta$ function of $\G$?
\end{q}
\begin{q}
	The abelian extension theory of $\Q_p$ corresponds to some topological properties of $p$-adic unitary operator, which connects the class field theory with functional analysis. What is the connection between non-commutative Iwasawa theory and $p$-adic spectral theory?
\end{q}
\begin{q}
	The geometry model of $p$-adic numbers and pro-finite groups are like some fractals. What is the definition of $p$-adic path-integral?
\end{q}

\bibliographystyle{plain}
\bibliography{Ref}

\end{document}